\documentclass[11pt,a4paper]{amsart}
\usepackage{ulem} 
\usepackage{times,epsfig}
\usepackage{amsfonts,amscd,amsmath,amssymb,amsthm}%
\usepackage{mathrsfs}
\usepackage{xcolor}
\numberwithin{equation}{section}

\DeclareMathOperator{\Ran}{Ran}

\newcommand{\R}{\mathbb{R}}

\newcommand{\C}{\mathbb{C}}
\newcommand{\Z}{\mathbb{Z}}
\newcommand{\N}{\mathbb{N}}

\newcommand{\fC}{\mathfrak{C}}

\newcommand{\cA}{{\mathcal A}}

\newcommand{\cH}{{\mathcal H}}

\newcommand{\cM}{{\mathcal M}}
\newcommand{\cN}{{\mathcal N}}
\newcommand{\cO}{{\mathcal O}}

\newcommand{\cR}{{\mathcal R}}

\newcommand{\rd}{\mathrm{d}}
\newcommand{\e}{\mathrm{e}}

{\bf}{\it}
{\bf}{\it}
{\bf}{\it}
{\bf}{\it}
{\bf}{\it}

\newtheorem{theorem}{Theorem}[section]{\bf}{\it}
\newtheorem{proposition}[theorem]{Proposition}{\bf}{\it}
{\bf}{\it}
\newtheorem{example}[theorem]{Example}{\it}{\rm}
\newtheorem{lemma}[theorem]{Lemma}{\bf}{\it}
\newtheorem{remark}[theorem]{Remark}{\it}{\rm}
\newtheorem{definition}[theorem]{Definition}{\bf}{\it}
{\bf}{\it}
{\bf}{\it}
{\bf}{\it}

\title{Reflection positivity in simplicial gravity}

\author[R. Schrader]{Robert Schrader}
\address{Robert Schrader\\ Institut f\"{u}r
Theoretische Physik, Freie Universit\"{a}t Berlin, \newline Arnimallee 14,
 D-14195 Berlin, Germany}
\email{robert.schrader@fu-berlin.de}
\date{\today File: \textbf{latgrav11.tex}}

\begin{document}

\begin{abstract}
Within the context of piecewise linear manifolds we establish reflection positivity with a Hilbert 
action given in terms of the Regge curvature and a cosmological term. Using this positivity a Hilbert space for 
a quantum theory is constructed and some field operators and observables are given. The set-up allows to introduce 
time reversal though no time exists.
All constructions are non-perturbative.
\end{abstract}

\maketitle

\section{Introduction}
The aim of this article is to prove reflection positivity in simplicial gravity.

First we give a short outline of the historical background.
Functional integrals provide an important non-perturbative tool in quantum physics. First applications 
in theories with finite degrees of freedom date back to Dirac \cite{Dirac} and to Feynman \cite{Feynman}.
Functional integrals give additional insight and provide new results with effective and short proofs. This holds 
in particular for Schr\"{o}dinger operators, for an account see {\it e.g.}
\cite{GlimmJaffe,Simon}. It was Symanzik \cite{Symanzik1,Syamnzik2}, who in the 60's strongly advertised the
use the theory of functional integrals in quantum field theories. It was taken up by Nelson \cite{Nelson} 
and its culmination was {\it constructive quantum field theory}, by which several relativistic quantum field theoretic 
models in 2 and 3 space-time dimensions could be constructed, see \cite{GlimmJaffe} for an overview. 
Another culmination originated from the seminal article by 
K.G. ~Wilson \cite{Wilson}, see also \cite{Polyakov}, who provided a lattice formulation of gauge and quark fields 
on the lattice. Prior to that F. Wegner \cite{Wegner} introduced many concepts including Wilson- and 
't Hooft-loops, however, only for 
the group $\Z_2$.
These articles improved our understanding of gauge theories, which was extended  
by detailed numerical calculations including Monte Carlo simulations.

Stimulated by the success in lattice gauge theories, in 1981 the proposal was made to combine 
Regge calculus \cite{Regge} with functional integration methods into a non-perturbative approximation 
scheme  for quantum gravity \cite{CMS1,F,RoWi}. Although Regge calculus originated as a concept in the geometry of 
piecewise linear (p.l.) spaces, it was Wheeler, who speculated on the possibility of employing it as 
a tool for constructing a quantum theory of gravity \cite{Wheeler}.
For this the names {\it lattice gravity} or 
{\it simplicial gravity} are often used, for overviews see e.g. \cite{Hamber,ReggeWilliams}. 
Meanwhile many numerical studies, including the renormalization group and the lattice continuum limit 
within Regge calculus have been carried out, see \cite{Hamber} sec. 8 and \cite{Hamber1}.

Now functional integration alone is not sufficient for obtaining a quantum theory. The addition 
ingredient is {\it reflection positivity}, originally formulated for Euclidean Green's functions (Schwinger functions). 
Reflection positivity is essentially necessary and sufficient to obtain a Hilbert space for a quantum field theory in the sense of Wightman 
\cite{Wightman, GardingWightman}. In \cite{Luescher,MenottiPelissetto, OsterwalderSeiler,Seiler} it was then 
shown that lattice gauge theories indeed satisfy reflection positivity. \cite{FILS} gives a systematic account of 
lattice models in statistical mechanics, which satisfy reflexions positivity.
Also the unitarity of certain representations of the {\it Virasoro algebra} and hence the existence of a Hilbert space
is equivalent to reflection positivity for the correlation function, see the comments in \cite{FQS}.

For more recent work on reflection 
positivity, sometimes in completely different contexts like Lie groups, see {\it e.g.} \cite{Jaffe,JaffeJanssens,Neeb-Olafsson}.

Next we explain the structure and the results of this article.
Informally a piecewise linear manifold is obtained by gluing together euclidean simplexes. So less informally
there is an underlying  simplicial complex and a set of lengths associated to the edges, which we call a metric on 
this simplicial complex. These lengths are subject to certain constraints, which turn out to be non-holonomic and 
which may be viewed as extensions of the triangle inequalities and of the condition that edge lengths have to be 
positive. We will only work with finite simplicial complexes. So Section \ref{sec:basics} serves as a 
preparation for the main result in that it recalls the main notions, definitions and known results, which will be 
used.
The question which conditions on a simplicial complex give rise to a {\it transfer matrix} will not be 
addressed. Also we will not consider the problems arising when 
taking the {\it thermodynamic limit}. Finally we will not consider the question of how to take 
the {\it scale limit}. 

In \cite{GlimmJaffe} p. 446 a whole section is devoted to reflection positivity (r.p.) on lattice approximations.
Mostly cubic lattices in $\R^d,2\le d\le 4$ are considered and there r.p. is preserved. 
In Regge calculus lattices are replaced by simplicial complexes and therefore r.p. has to be 
established anew. Thus the central definition in this article provides the geometric concept of a 
{\it reflection} in a particular class of finite 
pseudomanifolds, see Definition \ref{ref:def1} in Section \ref{sec:RP}. 

The main result on reflection positivity for a given reflection is then stated in Theorem \ref{theo:main}. 
A central point 
was to find suitable cut-offs, which do not spoil this geometric reflection and hence allow to prove reflection 
positivity. We recall the well known fact that in many models
reflection positivity holds on a formal level. But after the introduction of cut-offs, necessary in order 
to make interesting quantities finite (like the partition function), several seemingly sensible ones 
destroy reflection positivity. 
In addition, the problem of finding the appropriate action on lattices and satisfying reflection positivity 
is well known in lattice gauge theories, see the comments made in \cite{OsterwalderSeiler}.
The proposal we make on the cut-offs respect the constraints on the metric. 
It is satisfying that they include both a short distance and a long distance cut-off, 
corresponding intuitively to an ultraviolet and an infrared cut-off.
Having established reflection positivity, the construction of a (quantum) Hilbert space for simplicial gravity 
in Section \ref{sec:Hilbertspace} is then standard. Additional structure on the Hilbert space is obtained by
constructing field operators. 

\vspace{0.2cm}
\textbf{Acknowledgements.} The author thanks R. ~Flume, A. ~Jaffe, A. ~Knauf, and E. ~Seiler for 
helpful comments.

\section{Basic Notations and Definitions}\label{sec:basics}

In this section we briefly recall some definitions, notions and results, and which we shall make use of in the sequel.
 For more details on basic definitions and notations we refer to \cite{Spanier}.
\\
A {\it finite simplicial complex} $K$ consists of a finite set of elements called {\it vertices}
and a set of finite nonempty subsets of vertices called {\it simplexes} such that
\begin{enumerate}
\item{ Any set containing only one vertex is a simplex.}
\item{Any nonempty subset of a simplex is also a simplex.}
\end{enumerate}
A $j$-simplex will generally be denoted by $\sigma^j$. The {\it dimension} $j$ is the number of its
vertices minus 1. The 1-simplexes are called {\it edges}. 
If $\sigma^\prime\subseteq \sigma$, 
then $\sigma^\prime$ is called a {\it face}
of $\sigma$ and a {\it proper face} if $\sigma^\prime\neq\sigma$. 
We set $\dim K = \sup_{\sigma\in K}\dim \sigma$ and occasionally we shall write $K^n$ with 
$\dim K=n$, when we want to emphasize the dimension of $K$. 
A complex $L$ is called a
{\it subcomplex} of $K$ if the simplexes of $L$ are also simplexes of $K$. We then write $L \subseteq K$.

We shall almost exclusively consider special simplicial complexes given as follows. 
An $n-$dimensional {\it pseudomanifold} is a simplicial complex $K^n$ such that 
\begin{enumerate}
 \item{Every simplex is a face of some $n$-simplex.}
 \item{Every $(n-1)$-simplex is the face of at most two $n$-simplexes.}
 \item{If $\sigma$ and $\sigma^\prime$ are $n$-simplexes of $K^n$, there is a finite sequence 
$\sigma=\sigma_1,
\cdots, \sigma_m=\sigma^\prime$ of $n$-simplexes of $K$, such that $\sigma_i$ 
and $\sigma_{i+1}$ have an $(n-1)$-simplex in common.}
\end{enumerate}
Unless otherwise stated, the dimension $n$ will always be taken to be $\ge 3$.
The (possibly empty) boundary $\partial K^n$ of $K^n$ is the subcomplex formed by 
those $(n-1)$-simplexes, which lie in exactly one $n$-simplex, together with their faces. 
$\partial K^n$ is not necessarily a pseudomanifold. Given $K^n$ let $\Sigma^1(K^n)$ denote the set of 
$1-$simplexes in $K^n$.

By definition a {\it metric on $K^n$} is a map 
$$
\underline{z}:\;\Sigma^1(K^n)\quad \rightarrow\quad \R_+\;:\quad \sigma^1\;\mapsto\;z^{\sigma^1}
$$
satisfying the following properties. 
\begin{itemize}
\item{For each $k-$ simplex $\sigma^k\in K^n$, one can build a euclidean $k-$simplex 
$\sigma^k(\underline{z})$ whose edge lengths are given as $l^{\sigma^1}
=\sqrt{z^{\sigma^1}},\;\sigma^1\subseteq \sigma^k$.}
\end{itemize}
When these conditions are satisfied, $(K^n,\underline{z})$ is called a {\it piecewise linear (p.l.) space} of 
dimension $n$. Due to the third condition for a pseudomanifold, each $(K^n,\underline{z})$ is a connected space. 
For precise details, see \cite{CMS2}.
The analogous quantity in the continuum case is a pair $(M^n,g)$ where $M^n$ is a smooth 
$n-$dimensional manifold and
$g$ is any Riemannian metric on $M^n$. Further analogies between quantities 
in the p.l. and the continuum context will be provided below, and more are given in the first 
appendix of \cite{Schrader}.

Given $\underline{z}$, we denote by $|\sigma^k|(\underline{z})$ the volume of this euclidean $k-$simplex 
$\sigma^k(\underline{z})$, 
see also \eqref{amatrixvol} below. 
The {\it volume of $(K^n,\underline{z})$} is by definition 
\begin{equation}\label{volume}
V(K^n,\underline{z})=V(\underline{z})=\sum_{\sigma^n\in K^n}|\sigma^n|(\underline{z}).
\end{equation}

Let $(K^n,\underline{z})$ be given and let $\sigma^{n-2}\subset \sigma^n$. The {\it dihedral angle} 
$(\sigma^{n-2},\sigma^n)(\underline{z})$ at 
$\sigma^{n-2}(\underline{z})$ in $\sigma^n(\underline{z})$ may be defined as 
follows. Let the unique 
$\sigma^{n-1}_1, \sigma^{n-1}_2\subset \sigma^n$ 
be such that $\sigma^{n-2}=\sigma^{n-1}_1\cap\sigma^{n-1}_2$. In their realization as 
euclidean simplexes in $E^n$, let $\bf{n}_1$ and $\bf{n}_2$ be unit vectors, 
normal to $\sigma^{n-1}_1$ and $\sigma^{n-1}_2$ respectively and pointing outwards. 
Then the dihedral angle $0<(\sigma^{n-2},\sigma^n)<1/2$ (in units of $2\pi$) is defined as 
\begin{equation}\label{dihedral}
(\sigma^{n-2},\sigma^n)(\underline{z})=(\sigma^{n-2},\sigma^n)=\frac{1}{2}-
\frac{1}{2\pi}\arccos\langle \bf{n}_1, \bf{n}_2\rangle.
\end{equation} 
It is easy to see that this quantity is independent of the particular euclidean realization of 
$\sigma^{n-2}(\underline{z})$ and $\sigma^n(\underline{z})$.

The quantity 
\begin{align}\label{app:R}
\cR(K^n,\underline{z})=\cR(\underline{z})&=\sum_{\sigma^{n-2}\in K^n,\;\sigma^{n-2}\notin \partial K^n}
\left(1-\sum_{\sigma^n\in K^n,\;\sigma^n\supset\sigma^{n-2}}(\sigma^{n-2},\sigma^n)(\underline{z})\right)
|\sigma^{n-2}|(\underline{z})\\\nonumber
&\qquad\qquad +\sum_{\sigma^{n-2}\in \partial K^n}
\left(\frac{1}{2}-\sum_{\sigma^n\in K^n,\; \sigma^n\supset\sigma^{n-2}}
(\sigma^{n-2},\sigma^n)(\underline{z})\right)
|\sigma^{n-2}|(\underline{z})
\end{align}
is called the {\it Regge curvature of $(K^n,\underline{z})$}. It is analogous to the total scalar 
curvature in 
Riemannian geometry. In fact, in \cite{CMS2} a detailed account is given, how the Regge curvature 
approaches the total scalar curvature of a given Riemannian manifold $(M,g)$ through a sequence of 
p.l. spaces close to $(M,g)$.

By $\cM(K^n)$ we denote the set of all metrics on $K^n$. For a given ordering of the set 
$\Sigma^1(K^n)$ of edges in 
$K^n$,
$\cM(K^n)$ may be viewed as a subset of $\R_+^{n_1(K^n)}$, where $n_1(K^n)=|\Sigma^1(K^n)|$ is the 
number of edges in $K^n$. There is the following result, see \cite{Ge, Schrader}.
\begin{theorem}{\rm(Ge,Schrader)}\label{theo:convex}
 $\cM(K^n)$ is a non-empty, open convex cone.  
 \end{theorem}
\begin{theorem} {\rm (Schrader)}\label{theo:smooth}
Both $\cR$ and $V$ are smooth functions on 
$\cM(K^n)$.
\end{theorem}
The analogy to the smooth case is as follows. If $g_1$ and $g_2$ are two metrics on $M^n$, then 
$\lambda_1 g_1+\lambda_2 g_2\; (0<\lambda_1,0\le \lambda_2)$ is also a metric on $M^n$.

The proof of the first theorem, as given in \cite{Schrader}, is obtained from an insight into 
the geometric structure of $\cM(K^n)$, which we briefly explain now and which will turn out to be useful 
when we introduce cut-offs. 

First consider a euclidean $k$-simplex $\sigma^k$ in $E^k\;(1\le k)$, and label its $k+1$ vertices in an 
arbitrary order 
as $0,1,\cdots, k$. Assume the vertex $0$ is placed at the origin. We regard the other vertices as being
represented by the (linearly independent) vectors $v_i, 1\le i\le k$.
Then the length $l_{ij}=l_{ji}$ of the edge connecting the two different vertices $i$ and $j$ 
is given in the form 
\begin{align*}
z_{0j}&=l_{0j}^2=\langle v_j,v_j\rangle,\qquad\qquad \qquad 1\le j\le k,\\\nonumber
z_{ij}&=l_{ij}^2=\langle v_i- v_j,v_i-v_j\rangle,\qquad 1\le i,j\le k.
\end{align*}
We make the convention $z_{ii}=l_{ii}^2=0$.
As a consequence the $k\times k$ real, symmetric matrix $A=A(\underline{z})$ , 
$\underline{z}=\{z_{ij}\}_{0\le i,j\le k}$, with entries
\begin{equation}\label{amatrix}
 a_{ij}=a_{ji}=\langle v_i,v_j\rangle=\frac{1}{2}(z_{0i}+z_{0j}-z_{ij}),\qquad 1\le i,j\le k
\end{equation}
is positive definite. The volume of the euclidean $k$-simplex is obtained as
\begin{equation}\label{amatrixvol}
|\sigma^k|=|\sigma^k|(\underline{z})=\frac{1}{k!}\det A^{1/2}=
\frac{1}{k!}
\left(\langle v_1\wedge v_2\wedge \cdots v_k,
v_1\wedge v_2\wedge \cdots v_k\rangle\right)^{1/2}.
\end{equation}
Moreover $\det A(\underline{z})$ is independent of the particular choice of the ordering of the vertices.
It is a homogeneous polynomial in the $z_{ij}\;(i<j)$ of order $k$, totally symmetric in these variables.
For the particular case $k=2$ this relation gives the area of a triangle in terms 
of its edge lengths (squared), originally attributed to Heron of Alexandria. 

The converse is also true: Given a real, positive definite ( and hence symmetric) $k\times k$ matrix $A$, define 
$\underline{z}$ by $z_{0i}=a_{ii}$ and $z_{ij}= a_{ii}+a_{jj}-2a_{ij},\;(1\le i,j\le k)$. 
Then there is a euclidean $k-$simplex 
with edge lengths given as $l_{0i}=\sqrt{z_{0i}},\,l_{ij}=\sqrt{z_{ij}}$.

We will now use the following well known fact: The $k\times k$ symmetric matrix $B$ is positive definite if and 
only if $B$  
and all its quadratic sub-matrices along the diagonal have positive determinant, see {\it e.g.} 
\cite{Bhatia, TerrasII}.
By this discussion we obtain the following result: Let $\underline{z}(\sigma)
=\{z^{\sigma^1}\}_{\sigma^1\subseteq \sigma}$ and observe that $\underline{z}(\sigma^1)=z^{\sigma^1}$. 

The space of all metrics $\cM(K^n)$ on $K^n$ is given by a set of non-holonomic constraints in $\R_+^{n_1(K^n)}$
\begin{equation}\label{metricspace2}
\cM(K^n)=\left\{\;\underline{z}\in \R^{n_1(K^n)}\;|\; 
\det A(\underline{z}(\sigma))>0\;\mbox{for all}\;\sigma\in K^n,\, 1\le \dim \sigma\le n\;\right\}.
\end{equation}

\begin{remark}
These constraints are not independent of each other. In fact, let 
$$
\tau^1\subset\tau^2\subset\cdots\subset\tau^{n-1}\subset\tau^n
$$
be any increasing sequence of $n$ simplexes in $K^n$. If $\det A(\underline{z}(\tau^k))>0$ holds for all 
$1\le k\le n$ then 
$\det A(\underline{z}(\tau^\prime))>0$ for each simplex $\tau^\prime\; (\dim \sigma^\prime >0)$ in 
$\tau^n$.
\end{remark}

\section{Hilbert actions, volume forms on the space of all metrics and partition functions.}

With the total scalar curvature $\cR$ and the volume $V$ we may define the
{\it Hilbert action}
\begin{equation}\label{Hilbert}
H_{\gamma,\lambda}(\underline{z})= \gamma \cR(\underline{z})+\lambda V(\underline{z}),
\end{equation}
which is well defined everywhere on $\cM(K^n)$. The second term is just the familiar 
{\it cosmological term}. Let $\fC$ be any compact set in $\cM(K^n)$. Also let $\rd \mu$ be any volume 
form on $\fC$.
These data allow us to define the partition function
\begin{equation}\label{partition}
Z_{\fC}(\gamma,\lambda)=\int_{\fC}\e^{-H_{\gamma,\lambda}(\underline{z})}\rd\mu(\underline{z}). 
\end{equation}
The compactness of $\fC$ guarantees that the partition function is a finite and positive number.
In fact, due to the smoothness of $\cR(\underline{z})$ and $V(\underline{z})$ on $\cM(K^n)$ (see Theorem 
\ref{theo:smooth}), both these quantities are bounded on $\fC$.
In this article we shall use the following volume form 
\begin{equation}\label{volumeform}
 \rd\mu(\underline{z})=\prod_{\sigma^1}\rd z^{\sigma^1}.
\end{equation}
For a detailed discussion of this volume form, see \cite{HW2}. This volume form has been used 
extensively in computer simulations.
When the dimension equals $n=4$ it corresponds to the Misner measure 
\cite{FaddeevPopov, Faddeev,Misner}.

For comparison to be made below, we note that
\begin{equation}\label{dnu}
 \rd\nu =\frac{1}{Z_{\fC}(\gamma,\lambda)} \; \rd\mu
\end{equation}
is a probability measure on $\fC$. There is the associated space $L^2(\fC,\rd \nu)$ of square 
integrable functions, which becomes a Hilbert space via the scalar product
\begin{equation}\label{eucl1}
 \langle g,f\rangle_{L^2}=\int_{\fC} \overline{g(\underline{z})}f(\underline{z})\;\rd\nu(\underline{z}). 
\end{equation}
Furthermore mean values like ones for the Regge curvature and for the volume
\begin{equation}\label{L^2exp}
 \langle R \rangle =\int_{\fC} R(\underline{z})\; \rd\nu(\underline{z}) ,\quad 
\langle V \rangle=\int_{\fC} V(\underline{z})\; \rd\nu(\underline{z}).      
\end{equation}
have been studied extensively with numerical methods like Monte Carlo 
simulations, see \cite{Hamber}, sec. 8 and the references quoted there.
The relations
\begin{equation}\label{eucl3}
 \langle R \rangle=-\frac{\partial}{\partial \gamma}\ln Z_{\fC}(\gamma,\lambda),\quad
 \langle V\rangle=-\frac{\partial}{\partial \lambda}\ln Z_{\fC}(\gamma,\lambda)
\end{equation}
are familiar from statistical mechanics.

\section{Reflection positivity}\label{sec:RP}
In this section we formulate our main result. We start with some preparations.

\subsection{Reflections}
In this subsection we will study special automorphisms of simplicial complexes. In particular they will be 
involutive automorphisms. Let $\phi$ be such an involutive automorphism (i.e. $\phi=\phi^{-1}$) of
the simplicial complex $K$ and 
let $K^\prime$ be a subcomplex. Then the subcomplex $K^\prime\cup \phi K^\prime$, which {\it a priori} is 
possibly empty, is left invariant under $\phi$. 
Assume in addition that $\phi$ is the the identity on $K^\prime\cap \phi K^\prime$. By definition
the (simplicial) {\it gluing} of $K^\prime$ to $\phi K^\prime$ along $K^\prime\cap \phi K^\prime$, 
written as
\begin{equation}\label{ref:glue}
K^\prime\uplus_{K^\prime\cap \phi K^\prime}\phi K^\prime,
\end{equation}
is the simplicial complex obtained by first considering the disjoint union of $K^\prime$ and $\phi K^\prime$
and then by identifying each simplex $\sigma\in K^\prime\cap \phi K^\prime$, viewed as a simplex in $K^\prime$,
with the same simplex, now viewed as a simplex in $\phi K^\prime$. 
It is easy to prove that \eqref{ref:glue}
defines a subcomplex of $K$.
\begin{definition}\label{ref:def1}
An involutive automorphism of the pseudomanifold $K^n$, not equal to the identity, is called a {\rm reflection}, 
and written as $\vartheta$, if the following properties are 
 satisfied. There are two subcomplexes $K_+,K_-$ of $K$ such that 
 \begin{itemize}
 \item{$K_+$ and $K_-$ are pseudomanifolds of dimension $n$,}
 \item{$K_-=\vartheta K_+$, and hence $K_+=\vartheta K_-$ holds,}
 \item{The subcomplex $K_0=K_+\cap K_-$ is a non-empty $(n-1)-$dimensional pseudomanifold,}
 \item{$\vartheta$ is the identity on $K_0=K_+\cap K_-$,} 
 \item{$K=K_+\uplus_{K_+\cap K_-}K_-$.}
 \end{itemize}
 \end{definition}
$K_0$ is a subcomplex of both $K_+$ and $K_-$, which not necessarily is a pseudomanifold. Furthermore any $\sigma^{n-1}\in K_0$ is contained in exactly one 
$n-$simplex $\sigma^n\in K_+$ and one $n-$simplex $\sigma^{n\;\prime}\in K_-$. 
$\sigma^n\cap \sigma^{n\;\prime}=\sigma^{n-1}$ and 
$\sigma^{n\;\prime}=\vartheta \sigma^{n}$ hold.

To put this definition in analogy to the formulation of euclidean field theory as given in \cite{OS1,OS2}, 
intuitively 
$K_+$ is the {\it future}, $K_-$ is the {\it past} while $K_0=K_+\cap K_-$ is the {\it present}, 
so $\vartheta$ can be viewed as the (euclidean) time reflection. In view of this analogy we may say that 
although no time has been introduced (yet), at least a time direction has been singled out.
This observation will be elaborated on in Subsection \ref{subseq:timerev}. 

In Section \ref{sec:timezero} we shall 
return to this analogy.
\begin{example}\label{ex:1}
The pseudomanifold $K^n\;(n\ge 1)$ has two $n$-simplexes, denoted by $\sigma^n_+$ and $\sigma^n_-$. 
Their intersection is an $(n-1)-$simplex, denoted by $\sigma^{n-1}_0$. $K_\pm$ is the simplicial subcomplex 
consisting of $\sigma^{n-1}_\pm$ and its sub-simplexes. $K_0$ is the simplicial subcomplex 
consisting of $\sigma^{n-1}_0$ and its sub-simplexes such that $K_0=K_+\cap K_-$. Let $\sigma^0_\pm$ be the unique
vertex in $K_\pm$ not contained in $\sigma^{n-1}_0$. By definition $\vartheta$ is the automorphism of $K^n$, 
which is the identity on $K_0$ and which interchanges $\sigma^0_+$ and $\sigma^0_-$.
\end{example}
This example may be extended in the following way. Let $K^\prime$ be an $n-$dimensional pseudomanifold with non-empty
boundary $\partial K^\prime$. Let $K^\prime_0\subseteq \partial K^\prime$ be a nonempty subcomplex of 
$\partial K^\prime$. Let $K^{\prime\prime}$ be a copy of $K^\prime$ and denote 
by $K^{\prime\prime}_0$ the corresponding copy of $K^\prime_0$. Denote by $\iota$ the natural simplicial isomorphism
between $K^\prime$ and $K^{\prime\prime}$. Via $\iota$ we may identify the corresponding simplices in 
$K^\prime_0$ and $K^{\prime\prime}_0$. Call $K$ the resulting simplicial complex, a pseudomanifold of dimension $n$.
There are natural simplicial maps $\phi^\prime$ and $\phi^{\prime\prime}$ from $K^\prime$ and $K^{\prime\prime}$ 
respectively into $K$. Call the images $K_+$ and $K_-$. Moreover 
$\phi^\prime(K^\prime_0)=\phi^{\prime\prime}(K^{\prime\prime}_0)$, which we give the name $K_0$, a pseudomanifold of 
dimension $(n-1)$. Finally the reflection $\vartheta$ is given as follows. On $K_+$ the map $\phi^\prime$ is invertible.
Similarly on $K_-$ the map $\phi^{\prime\prime}$ is invertible. On $K_+$ the map $\vartheta$ equals 
$\phi^{\prime\prime}\circ\iota\circ \phi^{\prime^{-1}}$, while on $K_-$ the map $\vartheta$ equals 
$\phi^{\prime}\circ\iota^{-1}\circ \phi^{\prime\prime^{-1}}$. By construction $\vartheta(K_\pm)=K_{\mp}$.
It is easy to see $\vartheta=\vartheta^{-1}$ and that 
on $K_0=K_+\cap K_-$ the map $\vartheta$ thus defined is the identity map.

We omit the proof of the trivial 
\begin{lemma}\label{Hilbertinv}
 If $\vartheta$ is a reflection on $K^n$, then the relations $\cR(\vartheta\underline{z})=\cR(\underline{z}),\,
 V(\vartheta\underline{z})=V(\underline{z})$ and hence 
 $H(\vartheta\underline{z})=H(\underline{z})$ are valid.
 \end{lemma}
$\vartheta$ defines a smooth map from $\cM(K^n)$ into itself via 
$(\vartheta\underline{z})^{\sigma^1}=z^{\vartheta \sigma^{1}}$. For given $\underline{z}$ set 
\begin{equation}\label{zpm}
 \underline{z}_+=\pi_+\underline{z}=\left\{z^{\sigma^1}\right\}_{\sigma^1\in K_+}, 
 \quad \underline{z}_-=\pi_-\underline{z}=\left\{z^{\sigma^1}\right\}_{\sigma^1\in K_-}.
\end{equation}
such that $\vartheta \underline{z}_+=\underline{z}_-$.

\subsection{Reflection positivity}
In order to formulate reflection positivity 
we need to introduce cut-offs of $\cM(K^n)$. We will make the following choice
\begin{equation}\label{compact}
\fC_\kappa=\fC_\kappa(K^n)=\big\{\;\underline{z}\;|\: 
\det A(\underline{z}(\sigma))\ge 1/\kappa\;
\mbox{for all}\; \sigma\in K^n,\; \max(||\underline{z}_+||,||\underline{z}_-||)\;\le \kappa\big\}
\end{equation}
for any $\kappa$. By definition $\fC_\kappa\subset\fC_{\kappa^\prime} \subset\cM(K^n)$ for all 
$\kappa\le \kappa^\prime$ and each $\fC_\kappa$ 
is compact in $\cM(K^n)$. Their definition is motivated by the description \eqref{metricspace2} of 
$\cM(K^n)$.
It follows immediately from the definition that each $\fC_\kappa$ is reflection invariant. 
The $\fC_\kappa$ deserve to be called 
cut-offs, since they have the important property that they exhaust all of $\cM(K^n)$, that is 
$\cup_\kappa \fC_\kappa=\cM(K^n)$ holds.
Observe that each edge length is bounded and also bounded away from zero on each $\fC_\kappa$, since
$\sqrt{\kappa}\ge l^{\sigma^1}=\sqrt{z^{\sigma^1}}\ge 1/\sqrt{\kappa}$. So intuitively they provide both an 
infrared and an ultraviolet cut-off. In addition the volumes of each simplex are bounded from below
\begin{equation}
 |\sigma^k|(\underline{z})\ge \frac{1}{k!}\frac{1}{\sqrt{\kappa}},\quad \underline{z}\in \fC_\kappa.
\end{equation}
More importantly these cut-offs will not spoil reflection positivity, 
as we will see shortly. 

From now on, we will assume that the 1-simplexes of $K_+$ are ordered in some way, 
with the 1-simplexes in $K_0$ coming last. Via $\vartheta$ this induces an ordering on $K_-$. 
So we may view $\underline{z}_+$ as an element of $\R_+^{n_1(K_+)}$ and
$\pi_+:\underline{z}\:\mapsto \: \underline{z}_+$ as a map from $\fC_\kappa$ into $\R_+^{n_1(K_+)}$.
The set 
\begin{equation}\label{p+inv}
\fC_{+,\kappa}=\pi_+(\fC_\kappa),
\end{equation}
is easily seen to be compact. 
Let $C(\fC_{+,\kappa},\C)$ be the set of continuous, complex valued functions on $\fC_{+,\kappa}$. 
With the supremum norm
\begin{equation}\label{supnorm}
||f||_{\sup}=\sup_{\underline{z}\in \fC_{+,\kappa}}|f(\underline{z})|
\end{equation}
it is a commutative Banach algebra with unit, 
see e.g. \cite{Kaniuth} for definitions. The unit is just the constant function, equal to 1 
everywhere and will be denoted by $\e_\kappa$. Actually on continuous functions the essential supremum 
$||\cdot||_\infty$ and the supremum $||\cdot||_{\sup}$ agree. However, in order to avoid confusion, we will 
stick with the notation \eqref{supnorm}.

In the next section we shall exploit the fact that we are dealing 
with an algebra. In fact, we shall be able to construct what might be viewed as an analogue of 
a field algebra in quantum field 
theory. As for comparison with the continuum theory in gravity it corresponds to the family of functionals 
$\Phi(g)$ of the metric $g$.

We introduce the scalar product
\begin{equation}\label{scalarproduct}
 \langle f^\prime, f\rangle_\vee=\frac{1}{Z_{\fC_\kappa}(\gamma,\lambda)}
 \int_{\fC_\kappa}\overline{f^\prime((\vartheta\underline{z})_+)}f(\underline{z}_+)
 \e^{-H_{\gamma,\lambda}(\underline{z})}\rd \mu(\underline{z}).
 \end{equation}
 $\langle \e_\kappa, \e_\kappa\rangle_\vee=1$ is obvious. Observe that 
 $(\vartheta\underline{z})_+=\underline{z}_-$ holds.
 For given $K^n$ of course $\langle \cdot,\cdot\rangle_\vee$ 
 depends on the choice of $\vartheta,\gamma,\lambda$ and $\fC_\kappa$. This scalar product compares with the 
 $L^2$ scalar product \eqref{eucl1}.
\begin{lemma}\label{sesqui}
 $\langle\cdot,\cdot \rangle_\vee$ is a sesquilinear when the volume form is given by \eqref{volumeform}. 
 \end{lemma}
 \begin{proof}
 Anti-linearity in the first factor and linearity in the second factor in 
 $\langle f^\prime, f\rangle_\vee$ is clear. 
 It remains 
 to show $\overline{\langle f^\prime, f\rangle_\vee}=\langle f, f^\prime\rangle_\vee$. To see this, make the 
 variable transformation $\underline{z}^\prime=\vartheta\underline{z}$ in the expression \eqref{scalarproduct}. 
 Then use $\rd \mu(\underline{z})
 =\rd\mu(\underline{z}^\prime)$, viewed as a volume form on $\fC_\kappa$, as well as 
 $\vartheta \fC_\kappa=\fC_\kappa$ and 
 Lemma \ref{Hilbertinv}. So the claim follows.
 \end{proof}
 We omit the trivial proof of the next lemma.
 \begin{lemma}\label{contin}
  The map $(f_1,f_2)\;\mapsto\;\langle f_1, f_2\rangle_\vee$ from $C(\fC_{+,\kappa},\C)\times C(\fC_{+,\kappa},\C)$  
  into $\C$ is continuous. More precisely the estimate
  \begin{equation}\label{contin2}
  |\langle f_1, f_2\rangle_\vee|\le ||f_1||_{\sup}\;||f_2||_{\sup}
  \end{equation}
  holds.
\end{lemma}
As for reflection positivity the main result of this article is 
\begin{theorem}\label{theo:main}
Let $\vartheta$ be a reflection on the pseudomanifold $K^n$.  With the choice
\eqref{volumeform} of the volume form,  $\langle f, f\rangle_\vee\ge 0$ holds, that is 
$\langle\cdot,\cdot \rangle_\vee$ 
defines a scalar product, which possibly is degenerate.
\end{theorem}
The proof will be given in Appendix \ref{sec:app}.


\section{The quantum Hilbert space and some field operators and quantum observables.}\label{sec:Hilbertspace}
In this section $\fC_\kappa,\gamma$ and $\lambda$ will be fixed.

A Hilbert space, an basic ingredient in any model in quantum theory, can now be obtained as follows. We 
employ arguments used in \cite{OS1,OS2}:
We have to cope with the possibility that the scalar product may be degenerate, so we proceed as follows. 
Let the null space $\cN$ be the set of all $n\in C(\fC_{+,\kappa},\C)$ with $\langle n, n\rangle_\vee= 0$. 
$\cN$ may just be the null vector in $C(\fC_{+,\kappa},\C)$.

\begin{lemma}\label{nullspace}
$\cN$ is a closed linear space and $\eta\in\cN$ if and only if $\langle f,\eta\rangle_\vee=0$ for all 
$f\in C(\fC_{+,\kappa},\C)$. If $\eta$ is in $\cN$, then also its complex conjugate $\bar{\eta}$.
\end{lemma}
\begin{proof}
The proof trivially follows from the fact the fact Schwarz inequality also holds for the scalar product 
 $\langle\cdot,\cdot\rangle$. Closedness then follows from continuity, see Lemma \ref{contin}.
 The last claim follows trivially from the identity 
 $\langle\bar{\eta}, \bar{\eta}\rangle_\vee=\overline{\langle \eta,\eta \rangle_\vee}$.
 \end{proof}
 
The elements of the quotient space are just the cosets of $\cN$. Consider the map
\begin{align*}
\hspace{1.5cm}\psi \quad :\quad &C(\fC_{+,\kappa},\C) \quad \rightarrow\quad C(\fC_{+,\kappa},\C)
/\cN\\
&f\quad\qquad\qquad \mapsto\quad \psi(f)=f+\cN. 
\end{align*}
In the case that $\cN$ happen to be trivial, $\psi$ is chosen to be the identity map.

The map $\psi$ carries the positive semi-definite scalar product $\langle\cdot,\cdot \rangle_\vee$ into a positive 
definite scalar product on
$C(\fC_{+,\kappa},\C)/\cN$, denoted by $\langle\cdot, \cdot\rangle$. Thus we have
$$
\langle \psi(f^\prime), \psi(f)\rangle=\langle f^\prime, f\rangle_\vee.
$$
In particular the norm $||\cdot||$ obtained from the scalar product $\langle\cdot, \cdot\rangle$ has the property that
\begin{equation}\label{psicont}
 ||\psi(f)||=\left(\langle f, f\rangle_\vee\right)^{1/2}\le ||f||_{\sup}
\end{equation}
holds for all $f\in C(\fC_{+,\kappa})$. This gives
\begin{lemma}\label{lem:psicont}
 The map $\psi$ is continuous.
\end{lemma}
With this scalar product $\langle\cdot,\cdot \rangle$ the algebra $C(\fC_{+,\kappa},\C)/\cN$ is turned 
into a pre-Hilbert space $\cH^{pre}$ . 
Its (metric) completion is the desired Hilbert space $\cH$, which sometimes is called the {\it physical Hilbert space} 
in contrast to the Hilbert space $L^2(\fC_{+,\kappa},\rd \nu)$. By this construction $\cH^{pre}$ is automatically 
dense in $\cH$. $\cH$ depends on the parameters 
$\gamma,\lambda$, the cut-off $\fC_\kappa(K^n)$, the volume form $\rd \mu$ and the reflection $\vartheta$. Observe 
that in this construction the {\it ``future''} $K_+$ has been singled out, see the discussion after Definition 
\ref{ref:def1}. 

We set $\Omega_\kappa=\psi(\e_\kappa)\neq 0$ and call it {\it the vacuum}, or less flashy the {\it ground state}.
Since $\langle \e_\kappa,\e_\kappa\rangle_\vee=1$, $\Omega_\kappa$ is a unit vector. 
More generally the quantum gravity states are all the unit vectors in $\cH$, since we do not expect there are 
superselection rules \cite{StreaterWightman,WWW}.
In case the normalized states lie in the pre-Hilbert 
space, they are just functions of the metric. This is analogous to the continuum case, where one expects the states to be functions of 
the metric or rather on the associated moduli space, consisting of the set diffeomorphism classes of metrics. 
If $\cO$ is an observable and $\phi\in\cH$ is any state, then $\langle\phi,\cO\phi\rangle$ is called the 
expectation of $\cO$ in the state $\phi$. $\langle \Omega_\kappa, \cO \Omega_\kappa\rangle$ is called the 
vacuum expectation of $\cO$. 

The next aim is to construct a field algebra and interesting observables. Thus we want to define a 
{\it ``field operator'' $\Phi(f)$} 
acting on states $\psi(f^\prime)$ in the form 
\begin{equation}\label{desid}
\Phi(f)\psi(f^\prime)=\psi(f\,f^\prime) 
\end{equation}
with $f^\prime\in C(\fC_{+,\kappa},\C)$ and where $f$ is in some suitable function space $\cA$ such 
that the composition 
$f\,f^\prime$ is defined and is an element of $C(\fC_{+,\kappa},\C)$.
So the {\it desideratum} is 
$$
\Phi(f)\psi(f^\prime)=f\,f^\prime +f\cN\subseteq f\,f^\prime +\cN.
$$
For a corresponding discussion in euclidean field theory, 
see \cite{GlimmJaffe} p. 94.  

In order to define $\cA$, let $C_0(\cM(K^n),\C)$ be the algebra of all continuous functions on $\cM(K^n)$ with compact support. 
For given $K^n$ with reflection $\vartheta$ the algebra 
$C_0(\cM(K_0),\C)$, which is contained in $C_0(\cM(K^n),\C))$, is defined similarly. Note that all 
the algebras $C(\fC_{+,\kappa},\C)$ are $C_0(\cM(K^n),\C))$-modules and so {\it a fortiori} 
$C_0(\cM(K_0),\C)$-modules by defining the composition rule as
\begin{equation}\label{restr}
f\,f^\prime\doteq f|_{\fC_{+,\kappa}}f^\prime,
\end{equation}
where multiplication on the r.h.s. is in $C(\fC_{+,\kappa},\C)$.
As a consequence the algebras $C(\fC_{+,\kappa},\C)$ are not only modules but actually {\it ideals}.
So the following definition makes sense.
\begin{definition}\label{ass:1}
The algebra $\cA$ is defined to be $C_0(\cM(K^n),\C))$ if $\cN=\{0\}$ and 
$C_0(\cM(K_0),\C)$ if $\dim \cN>0$.
\end{definition}
We note that by Tietze's extension theorem the restriction map $f\:\rightarrow \:f|_{\fC_{+,\kappa}}$ is surjective 
from $C_0(\cM(K^n),\C))$ to $C(\fC_{+,\kappa},\C)$, see Appendix \ref{app:Tietze} for details.
\begin{proposition}\label{prop:module}
$\cN$ is an $\cA$-module. 
\end{proposition}
The proof will be given in Appendix \ref{2prop}.
Due to this result, $\Phi(f)$ with $f\in\cA$ is indeed an operator on $\cH$ with dense domain equal to 
$\cH^{pre}$, leaving this domain invariant and 
\begin{equation}\label{fieldop}
\Phi(f)\psi(f^\prime)=\psi(f\,f^\prime)
\end{equation}
holds. It gives the bound 
$||\Phi(f)\psi(f^\prime)||\le ||f||_{\sup}\,||f^\prime||_{\sup}$. 
This does \underline{not} imply that $\Phi(f)$ can 
be extended to a 
bounded operator on all of $\cH$. Nevertheless the relations 
$\Phi(\lambda_1 f_1+\lambda_2 f_2)=\lambda_1\Phi(f_1)+\lambda_2\Phi(f_2)$ and 
$\Phi(f_1\,f_2)=\Phi(f_1)\Phi(f_2)$ hold
on $\cH^{pre}$. From \eqref{fieldop} and some work using again the Tietze extension theorem, we conclude
\begin{proposition}\label{prop:cycl}
If $\cN=\{0\}$ holds, then the vacuum $\Omega_\kappa$ is a cyclic vector in $\cH$ for 
the set of operators $\Phi(f),\,f\in \cA$.
\end{proposition}
By the arguments used in the proof of Lemma \ref{sesqui}, it
is easy to see that
\begin{equation}
 \overline{\langle \psi(f^{\prime\prime}), \Phi(f)\psi(f^\prime)\rangle}
 =\langle \Phi(\bar{f})\psi(f^\prime),
\psi(f^{\prime\prime})\rangle. 
\end{equation}
So when $f$ is real and $f^{\prime\prime}=f^\prime$
we infer that 
$$
\langle \psi(f^\prime),\Phi(f)\psi(f^\prime)\rangle 
$$ 
is real. 

Here are some additional observables and we observe that by construction the {\it `` configuration space''} 
underlying $\cH$ is actually $K_+$.
The first important observables are obtained from $R(K_+,\underline{z}_+)$ and $V(K_+,\underline{z_+})$. For their 
relation to 
$R(K^n,\underline{z})$ and $R(K^n,\underline{z})$ see the appendix. 
These quantities are smooth functions of $\underline{z}_+$ on all $\cM(K_+)$. Via 
$\underline{z}_+=\pi_+\underline{z}$ we will view them as functions of 
$\underline{z}\in\cM(K^n)$. By restriction to $\fC_{+,\kappa}$ they define elements of $C(\fC_{+,\kappa},\C)$, denoted 
by the same symbol. Hence $\Phi(R(K_+,\underline{z}_+))$ and $\Phi(V(K_+,\underline{z}_+))$ are well defined as 
are their expectation values in states in $\cH^{pre}$. In particular one can consider the vacuum expectation values 
\begin{equation}\label{vacuum}
 \langle \Omega_\kappa, \Phi(R(K_+))\Omega_\kappa\rangle,\qquad 
\langle \Omega_\kappa, \Phi(V(K_+))\Omega_\kappa\rangle.
\end{equation}
As to be expected of a quantum theory, these (finite) expectations in the  vacuum do not vanish, 
In other words there are quantum fluctuations of both the curvature and the volume in the vacuum. 

In \cite{Schrader} we introduced and discussed the vector field given as the gradient of the Regge curvature. 
So the components of the vector field 
$\underline{\nabla}_{\underline{z}_+}R(K_+,\underline{z}_+)$ also define 
observables and expectations like 
$$
\langle \Omega_\kappa, \Phi(\underline{\nabla}_{\underline{z}_+}R(K_+))\Omega_\kappa\rangle.
$$
By the discussion in \cite{Schrader} the vector field 
$\underline{\nabla}_{\underline{z}_+}R(K_+,\underline{z}_+))$ can only become singular, when $\underline{z}_+$ tends 
to the boundary of $\cM(K_+)$ in such a way, that the volume of one (or more) of the $(n-2)-$simplexes tends to zero. 
Analogously the gradient of the volume can become singular, when the volume of one (or more) of the $n-$simplexes 
becomes small. This is another reason, why we introduced cut-offs in the form of the $\fC_\kappa(K^n)$.  

This discussion allows to consider {\it perturbations} in the following way. Assume for simplicity that $\cN=\{0\}$ 
such that $\cA=C_0(\cM(K^n),\C))$. Let $\e\in\cA$ be the function equal to 1 on $\cM(K^n)$. Consider $\Phi(\e+\delta f)$, 
where $\delta f$ is small. Thus $\Phi(\e+\delta f)\psi(f)=\psi(f)+\psi(\delta f|_{\fC_{\kappa}} f)$. 
For the special case when $f=\e_{\kappa}$ this describes perturbations of the vacuum $\Omega_\kappa$
$$
\Phi(\e+\delta f)\Omega_{\kappa}=\Omega_\kappa+ \psi(\delta f|_{\fC_{\kappa}}\e_\kappa).
$$

We conclude this section by providing an analogue of the functional derivative
$$
\frac{\delta}{\delta g_{ij}(x)}
$$
in the continuum case and which has an interpretation as an infinite dimensional gradient. 
As a first example we apply this gradient to the total scalar 
curvature on $(M^ng)$.
$$
R(g)=\int R(g)(x)\sqrt{\det g_{ij}(x)}\,\rd vol(g)(x),
$$
where
$$
\rd vol(g)(x)=\sqrt{\det g_{ij}(x)}\;\rd x^1\wedge\rd x^2\cdots\wedge\rd x^n,
$$
is the volume form. Its functional derivative is (minus) the Einstein tensor
$$
\frac{\delta R(g)}{\delta g_{ij}(x)}=-\left(
 Ric(g)^{ij}(x)-\frac{R(g)(x)}{2}g^{ij}(x)\right).
$$
On the level of quadratic forms we can actually do more. Formally 
\begin{align*}
\langle \psi(g_1), R(K_-)\psi(g_2)\rangle &=\frac{1}{Z_{\fC_\kappa}(\gamma,\lambda)}
 \int_{\fC_\kappa}  \overline{g}_1(\underline{z}_-)
 R(K_-,\underline{z}_-)
 g_2(\underline{z}_+)
 \e^{-H_{\gamma,\lambda}(\underline{z})}\rd \mu(\underline{z})\\\nonumber
 &=\frac{1}{Z_{\fC_\kappa}(\gamma,\lambda)}\int_{\fC_\kappa} R(K_-,\underline{z}_-)
 \overline{g}_1(\underline{z}_-)
 \e^{-H_{\gamma,\lambda}(\underline{z})}\;g_2(\underline{z}_+) \;\rd \mu(\underline{z})\\\nonumber
 &=\langle \psi(\overline{g}_2),\Phi(R(K_+))\psi(\overline{g}_1)\rangle\\\nonumber
 &=\overline{\langle \psi(g_2),\Phi(R(K_+))\psi(g_1)\rangle}\\\nonumber
 &=\langle \Phi(R(K_+))\psi(g_1),\psi(g_2)\rangle,
\end{align*}
where for the third equality we have used arguments similar to the ones uses for the proof of Lemma \ref{sesqui}.
In particular this gives 
\begin{equation}\label{expR}
\langle \psi(g),\Phi(R(K^n))\psi(g)\rangle
=2\Re e(\langle \psi(g),\Phi(R(K_+))\psi(g)\rangle).
\end{equation}
Moreover 
\begin{equation}\label{expR1}
\langle \psi(g), \Phi(R(K^n))\psi(g)\rangle= 2\langle \psi(g),\Phi(R(K_+))\psi(g)\rangle
\end{equation}
holds whenever $g$ is real, $\overline{g}=g$.

Quantities like $V(K^n)$ and $\Phi(\underline{\nabla}_{\underline{z}_+}R(K))$ can be discussed similarly.

As a second example consider the special functional
\begin{equation}\label{psig}
 \Psi(g)=\int_{M^n}\Psi^{ij}(x)g_{ij}(x)\,\rd vol(g)(x). 
\end{equation}
$\Psi^{ij}(x)$ is a symmetric tensor field on $M^n$.
Its functional derivative is
$$
\frac{\delta \Psi(g)}{\delta g_{ij}(x)}=\Psi^{ij}(x)+\frac{1}{2}\Psi^{kl}(x)g_{kl}(x)g^{ij}(x).
$$
In order to formulate its p.l. version, let $C^1(\fC_{+,\kappa},\C)$ be the subalgebra consisting of functions, 
whose first order partial derivatives are all in $C(\fC_{+,\kappa},\C)$. 
\begin{lemma}\label{lem:Tietze}
$C^1(\fC_{+,\kappa},\C)$ is dense in $C(\fC_{+,\kappa},\C)$
\end{lemma}
By this lemma and by Lemma \ref{lem:psicont} the set $\psi(C^1(\fC_{+,\kappa},\C))$ is dense in $\cH$. 
The proof uses standard and well known techniques. For completeness and for the convenience of the reader, 
however, we will nevertheless present a detailed proof in Appendix \ref{app:Tietze}.
\begin{definition}
 The set of operators $\underline{\nabla}^{op}$ on $\cH$ with dense domain 
 $\psi(C^1(\fC_{+,\kappa},\C)$ is defined by 
 $$
 \underline{\nabla}^{op}\psi(f)=\psi(\underline{\nabla}_{\underline{z}_+}f),\qquad f\in C^1(\fC_{+,\kappa},\C). 
 $$
\end{definition}
An easy calculation using the Leibniz rule gives the 
\begin{proposition}
The operator commutation relations
$$
\left[\,\underline{\nabla}^{op}, \Phi(f)\,\right] =\Phi(\underline{\nabla}_{\underline{z}_+}f), \qquad 
f\in C^1(\fC_{+,\kappa},\C) 
$$
are valid on the domain $\psi(C^1(\fC_{+,\kappa},\C))$, a dense set in $\cH$.
\end{proposition}
The following example compares with $\Psi(g)$ in the smooth case, see \eqref{psig}.
\begin{example}\label{ex:psig}
 Let $f(\underline{z}_+)=\langle \underline{a},\underline{z}_+\rangle$ with arbitrary 
 $\underline{a}\in\R^{n_1(K_+)}$ . Then $\underline{\nabla}_{\underline{z}_+}f$ is the constant 
 vector field $\underline{a}$ and $\Phi(\underline{\nabla}_{\underline{z}_+}f)$ is multiplication by 
 $\underline{a}$ on $\cH$. More generally for $f=\langle \underline{a},\underline{z}_+\rangle^p,\;p\in \N$,
 $\Phi(\underline{\nabla}_{\underline{z}_+}f)$ is multiplication by 
 $p\langle \underline{a},\underline{z}_+\rangle^{p-1}\underline{a}$.
\end{example}

\vspace{0.3cm}

\subsection{Time reversal}\label{subseq:timerev}

As promised we now elaborate on the time reflection.
Introduce the scalar product 
$$
\langle g^\prime,g\rangle_{\vee,-}=\frac{1}{Z_{\fC_\kappa}(\gamma,\lambda)}
 \int_{\fC_\kappa}\overline{g^\prime((\vartheta\underline{z})_-)}g(\underline{z}_-)
 \e^{-H_{\gamma,\lambda}(\underline{z})}\rd \mu(\underline{z}).
$$
on $C(\fC_{-,\kappa},\C)\times C(\fC_{-,\kappa},\C)$. Define the following antilinear map $\Theta$ from
$C(\fC_{+,\kappa},\C)$ onto $C(\fC_{-,\kappa},\C)$ by
$$
(\Theta f)(\underline{z}_-)=\overline{f((\vartheta\underline{z})_+)}
$$
with inverse 
$$
 (\Theta^{-1} g)(\underline{z}_+)=\overline{g((\vartheta\underline{z})_-)}.    
$$
A direct comparison with \eqref{scalarproduct} gives 
\begin{equation}\label{scprtheta}
\langle \Theta f^\prime, \Theta f\rangle_{\vee,-}=\langle f, f^\prime\rangle_{\vee}. 
\end{equation}
In analogy to the definition of $\cN$ we set
$$
\cN_-=\left\{g\;|\; \langle g, g\rangle_{\vee,-}=0\right\}.
$$
Lemma \ref{nullspace} for $\cN$ carries over to $\cN_-$.
 On the coset space 
$C(\fC_{-,\kappa},\C)/\cN_-$ we obtain a positive definite sesquilinear form making this space a pre-Hilbert 
space $\cH^{pre}_-$. We denote by $||\cdot||_-$ the Hi;bert space norm, that is 
$||g||_-=(\langle g,g\rangle)^{1/2}$. Consider the map 
\begin{align*}
\hspace{1.5cm}\psi_- \quad :\quad &C(\fC_{-,\kappa},\C) \quad \rightarrow\quad C(\fC_{-,\kappa},\C)
/\cN_-\\
&g\quad\qquad\qquad \mapsto\quad \psi_-(g)=g+\cN_-. 
\end{align*}
Since $\overline{\cN}=\cN$ it follows easily that $\Theta \cN=\cN_-=\overline{\cN}_-$. 
Therefore it makes sense to set
$$
\Theta\psi(f)=\psi_-(\Theta f)
$$
Combined with \eqref{scprtheta} this makes $\Theta$ an antilinear map from $\cH^{pre}$ onto $\cH_-^{pre}$. 
This map is isometric, that is $||\Theta f||_-=||f||$ holds. Therefore $\Theta$ extends to an 
antiunitary map from $\cH$ onto $\cH_-$, the Hilbert space completion of $\cH^{pre}_-$.
For comparison recall that in non-relativistic quantum mechanics time reversal also involves taking the complex 
conjugate by which time reversal becomes an antiunitary operator.

\vspace{0.3cm}


\subsection{The {\it ``time zero''} Hilbert space and field operators.}\label{sec:timezero}~~\\

The Banach algebra $C(\pi_0(\fC_{+,\kappa}),\C)$ may be viewed as a closed subalgebra of 
$C(\fC_{+,\kappa},\C)$.
In fact an element $f_0\in C(\pi_0(\fC_{+,\kappa}),\C)$ may be viewed as function of 
$\pi_0 \underline{z}\in \pi_0(\fC_{+,\kappa})$ when 
$\underline{z}\in \fC_{+,\kappa}$. With this understanding the scalar product $\langle\cdot,\cdot\rangle_\vee$ 
reduces to a scalar product on $C(\pi_0(\fC_{+,\kappa}),\C)\times C(\pi_0(\fC_{+,\kappa}),\C)$, which we denote by 
$\langle\cdot,\cdot\rangle_{\vee,0}$.

\begin{proposition}\label{vee0}
The scalar product $\langle\cdot,\cdot\rangle_{\vee,0}$ is positive definite.
\end{proposition}
This proposition states that $\cN\cap C(\pi_0(\fC_{+,\kappa}),\C))=\{0\}$.
The proof will be given Appendix \ref{2prop}.

Let $\cH_0$ be the Hilbert space completion of $C(\pi_0(\fC_{+,\kappa}),\C)$, the latter being viewed 
as a pre-Hibert space with norm obtained from the scalar product
$\langle\cdot,\cdot\rangle_{\vee,0}$. 

$\cH_0$ is a closed subspace of $\cH$. In fact, the restriction to $\cH_0$ 
of the scalar product $\langle\cdot,\cdot\rangle$ on $\cH$ is just $\langle\cdot,\cdot\rangle_{\vee,0}$.
We do not know whether these spaces agree.
Because $\e_\kappa\in C(\pi_0(\fC_{+,\kappa}),\C))$ the vector $\Omega_\kappa$ lies in $\cH_0$.
Let $\cA_0=C(\cM(K_0))$. By repeating the arguments used after \eqref{desid} and observing that 
$\cN\cap C(\pi_0(\fC_{+,\kappa}),\C))=\{0\}$, $C(\pi_0(\fC_{+,\kappa}),\C))$ is an ideal in $\cA_0$.
Therefore $\Phi_0(f_0)$ with $f_0\in\cA_0$ is well defined as an densely defined operator by setting 
$$
\Phi_0(f_0)\psi(f_0^\prime)=\psi(f_0\,f_0^\prime) \quad f_0^\prime\in C(\pi_0(\fC_{+,\kappa}),\C)).
$$
The vacuum $\Omega_\kappa$ is a cyclic vector in $\cH_0$ for 
the family $\Phi_0(f_0),\,f_0\in\cA_0$. In particular the Regge curvatures for the 
metric spaces $(K_0,\underline{z}_0)$ define the observable $\Phi_0(R(K_0))$. Similarly one obtains an 
observable for the volume.

In contrast to the construction of $\cH$ no time direction is singled out, that is for the construction 
of 
$\cH_0$ the reflection of $\vartheta$ is not needed. However, time reflection on the level 
of Hilbert space can still be formulated. In fact, $\Theta$ restricted to the pre-Hilbert space is just complex 
conjugation, $(\Theta f_0^\prime)(\underline{z}_0)= \overline{f_0^\prime(\underline{z}_0)}$. 
So $\Theta$ extends to an
antiunitary map of $\cH_0$ onto itself with $\Theta^2=id$.

\begin{remark}
In relativistic quantum field theory one usually considers fields as operator valued distributions satisfying the 
G\r{a}rding-Wightman axioms, see \cite{StreaterWightman,GardingWightman}.
Thus one deals with smeared-out fields in $d$ space-time dimensions, for bosonic fields they are of the typical form
$$
\int_{\R^d}\phi(\vec{x},t)f(\vec{x}, t)\;\rd^{d-1} \vec{x}\;\rd t,
$$
where $f$ is in Schwartz space.
Also (smeared-out) time zero fields 
$$
\int_{\R^{d-1}}\phi(\vec{x},0)f(\vec{x})\;\rd^{d-1}\vec{x},
$$
have been studied for some time. They have the conceptual disadvantage, that they do not exhibit 
relativistic covariance directly. Free time zero bosonic fields have been analyzed in 
detail and there the time zero Hilbert space agrees with the full Hilbert space. They also exist as operators 
in the two dimensional model $P(\phi)_2$ and as quadratic forms in $\phi^4_2$ theories,
see {\it e.g.} \cite{GlimmJaffe}. 
\end{remark}


\begin{appendix}
\section{Proof of Theorem \ref{theo:main}}\label{sec:app}
  
First we extend the notation introduced in \eqref{zpm}.  
For given $\underline{z}$ write 
$$
\underline{z}_+=\left(\widetilde{\underline{z}}_+,\underline{z}_0\right),\qquad
\underline{z}_-=\left(\widetilde{\underline{z}}_-,\underline{z}_0\right)
$$
with
$$
\widetilde{\underline{z}}_\pm=\widetilde{\pi}_\pm\underline{z}=
\left\{z^{\sigma^1}\right\}_{\sigma^1\in K_\pm,\;\sigma^1\notin K_0},\qquad
\underline{z}_0=\pi_0\underline{z}= \left\{z^{\sigma^{1}}\right\}_{\sigma^1\in K_0},
$$
such that 
$$
\underline{z}=\left(\widetilde{\underline{z}}_+,\underline{z}_0,
\widetilde{\underline{z}}_-\right),\quad \vartheta\underline{z}=
(\underline{z}_-,\widetilde{\underline{z}}_+)
$$
and with the ordering introduced above. Set
\begin{equation}\label{tildeC}
\widetilde{\fC}_{+,\kappa}=\widetilde{\pi}_+(\fC_{+,\kappa})=
\Big\{\widetilde{\pi}_+\underline{z}\:\Big|\:\underline{z}\in \fC_{+,\kappa}\; \Big\}.
\end{equation}
$\widetilde{\fC}_{-,\kappa}$ is defined analogously.

Obviously 
$(\widetilde{\underline{z}}_+,\underline{z}_0)\in\cM(K_+),\;(\widetilde{\underline{z}}_-,\underline{z}_0)\in \cM(K_-)$ 
and $\underline{z}_0\in\cM(K_0)$ when $\underline{z}\in\cM(K^n)$. 
More importantly the following converse is valid. In the next two lemmas we change the notation slightly and write
$\underline{w}_-=(\underline{w}_0,\widetilde{\underline{w}}_-)$.
\begin{lemma}\label{lem:conv}
Assume $\underline{w}_+\in\cM(K_+)$ and 
$\underline{w}_-\in\cM(K_-)$, such that $\pi_0\underline{w}_+=\pi_0\underline{w}_-$ holds and hence  
$\underline{w}_0\doteq\pi_0\underline{w}_+\in\cM(K_0)$ is valid. Then 
$$
\underline{w}=(\widetilde{\underline{w}}_+,\underline{w}_-)=(\widetilde{\underline{w}}_+,\underline{w}_0,\widetilde{\underline{w}}_-)=
(\underline{w}_+,\widetilde{\underline{w}}_-)
\in\cM(K^n).
$$
\end{lemma}
The preceding lemma carries over to the cut-off situation in the following way. 
\begin{lemma}\label{lem:conv1}
If $\underline{w}_\pm\in\fC_{\kappa}(K_\pm)$ with $\pi_0 \underline{w}_+=\pi_0 \underline{w}_-$ then 
$\underline{w}=(\widetilde{\underline{w}}_+, \underline{w}_-)
=(\underline{w}_+,\widetilde{\underline{w}}_-)\in\fC_\kappa(K^n)$.
 $\pi_0$ restricted to $\fC_\kappa(K^n)$ is a map into $\fC_\kappa(K_0)$, which is not necessarily surjective.
\end{lemma}
The proofs of these two lemmas are trivial and will be omitted.
In what follows $\kappa$ will be fixed. Decompose the curvature $\cR$ and the volume $V$ as follows
 \begin{align}\label{app:RVdecomp1}
  \cR(\underline{z})&=\cR_+(\underline{z})+\cR_-(\underline{z})\\\nonumber
  V(\underline{z})&=V_+(\underline{z})+V_-(\underline{z})
 \end{align}
with
\begin{align}\label{app:RVdecomp2}
\cR_\pm(\underline{z})&=\quad\sum_{\sigma^{n-2}\in K_\pm\setminus\partial K_\pm}
\left(1-\sum_{\sigma^n\in K_\pm} (\sigma^{n-2},\sigma^{n})(\underline{z})\right)|\sigma^{n-2}|(\underline{z})
\\\nonumber
&\qquad\qquad +\sum_{\sigma^{n-2}\in \partial K\pm}
\left(\frac{1}{2}-\sum_{\sigma^n\in K_\pm}(\sigma^{n-2},\sigma^n)(\underline{z})\right)
|\sigma^{n-2}|(\underline{z}),\\\nonumber
V_\pm(\underline{z})&=\sum_{\sigma^n\in K_\pm}|\sigma^n|(\underline{z}).
\end{align}
Observe that by definition of a reflection there is no $n-$simplex contained in $K_0$, so the last sum is 
actually only 
over those $\sigma^n$ which are in $K_\pm\setminus K_0$. Similarly the sums inside the two braces are only over 
those $\sigma^n$ which are in $K_\pm\setminus K_0$. To sum up, both $R_\pm$ and $V_\pm$ actually depend only 
on $\underline{z}_\pm$. With this understanding, it is easy to convince oneself 
that $ \cR_+(K^n,\underline{z}_+)=\cR_+(\underline{z}_+)=\cR(K_+,\underline{z}_+)$.

We note that $\vartheta (K_\pm\setminus K_0)=K_\mp\setminus K_0$.
Obviously  under reflections $\cR_\mp(\underline{z})=\cR_\pm(\vartheta\underline{z})$ and 
$V_\mp(\underline{z})=V_\pm(\vartheta\underline{z})$.

Correspondingly $H_{\gamma,\lambda}(\vartheta\underline{z})=H_{\gamma,\lambda}(\underline{z})$ holds.
With the decomposition $H_{\gamma,\lambda}(\underline{z})=H_{+,\gamma,\lambda}(\underline{z})
+H_{-,\gamma,\lambda}(\underline{z})$ the relations $H_{\pm,\gamma,\lambda}(\vartheta\underline{z})=
 H_{\mp,\gamma,\lambda}(\underline{z})$ and 
\begin{equation}\label{app:Hdecomp}
 \e^{-H_{\gamma,\lambda}(\underline{z})}
 =\e^{-H_{-,\gamma,\lambda}(\underline{z})}\;\e^{-H_{+,\gamma,\lambda}(\underline{z})}
\end{equation}
are valid.

We are now prepared to complete the proof of the theorem. 
Observe that 
$$
\e^{-H_{\pm,\gamma,\lambda}(\underline{z})}
$$
is actually only dependent on $\underline{z}_\pm$. Therefore it makes sense to set
$$
\widetilde{f}(\underline{z}_+)=\e^{-H_{+,\gamma,\lambda}(\underline{z})}\,f(\underline{z}_+),\qquad
\widetilde{f}((\vartheta\underline{z})_+)
=\e^{-H_{-,\gamma,\lambda}(\underline{z})}\,f((\vartheta\underline{z})_+).
$$
For any $\underline{z}_0\in\pi_0 \fC_\kappa(K^n)$
define the non-empty sets $\widetilde{\fC}_{\kappa,\,\pm}(\underline{z}_0)
=\widetilde{\pi}_\pm\pi_0^{-1}\{\underline{z}_0\}$.
Then we obtain
\begin{align*}
&\langle f,f\rangle=\\\nonumber
&\int\limits_{\pi_0 \fC_\kappa(K^n)}\prod_{\sigma^1\in K_0} \rd z^{\sigma^1}
\left(\:\int\limits_{\widetilde{\fC}_{\kappa,\,-}(\underline{z}_0)}\overline{\widetilde{f}(\underline{z}_-)}
\prod_{\sigma^1\in K_-\setminus K_0} \rd z^{\sigma^1}\right)
\left(\:\int\limits_{\widetilde{\fC}_{\kappa,\,+}(\underline{z}_0)}\widetilde{f}(\underline{z}_+)
\prod_{\sigma^1\in K_+\setminus K_0}\rd z^{\sigma^1}\right).
\end{align*}
The terms in both braces are functions of $\underline{z}_0$ only. Moreover, since 
$\vartheta \widetilde{\fC}_{\kappa,+}(\underline{z}_0)=\widetilde{\fC}_{\kappa,-}(\underline{z}_0)$ or more precisely 
$\vartheta\widetilde{\pi}_+\underline{z}=\widetilde{\pi}_-\underline{z}$ for 
$\underline{z}\in\pi^{-1}_0\{\underline{z}_0\}$,
these terms are the complex conjugate of each other. Carrying out the 
last integral over $\pi_0\fC_{\kappa}(K^n)=\Ran \pi_0\subset \fC_\kappa(K_0)$ concludes 
the proof of the theorem.


\section{Proof of Propositions \ref{prop:module} and \ref{vee0}.}\label{2prop}

\subsection{Proof of Proposition \ref{prop:module}}~~\\
By Lemma \ref{nullspace} it suffices to prove $\langle f\eta,f\eta\rangle_\vee=0$
for $f\in\cA$ and $\eta\in\cN$.
If $\cN$ is trivial, $\cN=\{0\}$, there is nothing to prove. Otherwise by definition 
$\cA=C_0(\cM(K_0),\C)$.  By \eqref{scalarproduct}
$$
\langle f\eta,f\eta\rangle_\vee=\frac{1}{Z_{\fC_\kappa}(\gamma,\lambda)}
 \int_{\fC_\kappa}\overline{\eta((\vartheta\underline{z})_+)}\eta(\underline{z}_+)
 |f(\pi_0\underline{z})|^2
 \e^{-H_{\gamma,\lambda}(\underline{z})}\rd \mu(\underline{z}).
$$
With 
\begin{align*}
F(\underline{z}_0)&\doteq\int_{\widetilde{\fC}_{+,\kappa}}\eta(\widetilde{\underline{z}}_+)
\e^{-H_+(\gamma,\lambda)(\pi_+\underline{z})}
\prod_{\sigma^1\in K_+\setminus K_0}\rd z^{\sigma^1}\\\nonumber
&=\int_{\widetilde{\fC}_{-,\kappa}}\eta(\widetilde{\underline{z}}_-)
\e^{-H_-(\gamma,\lambda)(\pi_-\underline{z})}
\prod_{\sigma^1\in K_-\setminus K_0}\rd z^{\sigma^1},
\end{align*}
where the sets $\widetilde{\fC}_{\pm,\kappa}$ 
are defined in \eqref{tildeC}, we obtain 
$$
\langle f\eta,f\eta\rangle_\vee=\frac{1}{Z_{\fC_\kappa}(\gamma,\lambda)}\int_{\pi_0 \fC_\kappa}
\prod_{\sigma^1\in K_0} \rd z^{\sigma^1}
|F(\underline{z}_0)|^2
|f(\pi_0\underline{z})|^2.
$$
But since $\eta\in\cN$ the continuous function $F(\underline{z}_0)$ vanishes everywhere.

\subsection{Proof of Proposition \ref{vee0}.}~~\\
In view of \eqref{scalarproduct} we have 
\begin{equation}\label{hilbert0}
\langle f_0,f_0\rangle_{\vee,0}=  \frac{1}{Z_{\fC_\kappa}(\gamma,\lambda)}
 \int_{\fC_\kappa}|f_0(\pi_0\underline{z})|^2
 \e^{-H_{\gamma,\lambda}(\underline{z})}\rd \mu(\underline{z}).
\end{equation}
Set  
\begin{align*}
G(\pi_0\underline{z})&\doteq\int_{\widetilde{\fC}_{+,\kappa}}\e^{-H_+(\gamma,\lambda)(\pi_+\underline{z})}
\prod_{\sigma^1\in K_+\setminus K_0}\rd z^{\sigma^1}\\\nonumber
&=\int_{\widetilde{\fC}_{-,\kappa}}\e^{-H_-(\gamma,\lambda)(\pi_-\underline{z})}
\prod_{\sigma^1\in K_-\setminus K_0}\rd z^{\sigma^1}\\\nonumber
&\,>0,
\end{align*}
where the second equality follows from reflection invariance. Again the sets 
$\widetilde{\fC}_{\pm,\kappa}$ 
are defined in \eqref{tildeC}.
We insert $G$ into \eqref{hilbert0} and obtain
$$
\langle f_0,f_0\rangle_{\vee,0}=  \frac{1}{Z_{\fC_\kappa}(\gamma,\lambda)}
 \int_{\fC_\kappa}|f_0(\pi_0\underline{z})|^2\;G(\pi_0\underline{z})^2
 \rd \mu(\underline{z})\ge 0
$$
which vanishes if and only if $f_0=0$.


\section{Proof of Lemma \ref{lem:Tietze}}\label{app:Tietze}

By decomposing any complex valued function into its real and imaginary part, it suffices to prove that 
$C^1(\fC_{+,\kappa},\R)$ is dense in $C(\fC_{+,\kappa},\R)$.
In order not to burden the notation unnecessarily, from now on 
$\underline{x},\underline{y}$ will denote points in $\R^{\bar{n}}$ with $\bar{n}=n_1(K_+)$.
Consider any $f\in C(\fC_{+,\kappa},\R)$. We will invoke the Tietze extension lemma, see {\it e.g.} \cite{Jaenich}, 
by which $f$ has a continuous extension $F$ to
all of $\R^{\bar{n}}$, that is $F|_{\fC_{+,\kappa}}=f$. This is possible, since each $\R^k$ is a normal space.
Also $-||f||_{\sup}\le F\le ||f||_{\sup}$ holds.
We claim $F$ can be assumed to have compact support. In fact, in case $F$ is not already of compact support, choose a 
smooth function $0\le h$ on $\R$ with support in the interval $[0,1]$ satisfying 
$$
\int_{0}^1 h(r)\rd r=1.
$$
With  
$$
R_0=\sup_{\underline{x}\in \fC_{+,\kappa}}|\underline{x}|
$$
define the smooth function
$$
\chi(\underline{x})=1-\int_{0}^{|\underline{x}|-R_0-1}\;h(r)\;\rd r.
$$
It satisfies $\chi|_{\fC_{+,\kappa}}=1$, such that $(\chi\,F)|_{\fC_{+,\kappa}}=f$ holds, and the claim follows. 
So from now on $F$ will be assumed to have compact support.
Next choose a smooth function $0\le g(\underline{x})$ with support in $|\underline{x}|\le 1$ satisfying
$$
\int_{\R^{\bar{n}}}\,g(\underline{x}) \,\rd^{\bar{n}}\underline{x}=1
$$
and then set
$$
g_{\varepsilon}(\underline{x})=\varepsilon^{-\bar{n}}g(\varepsilon^{-1}\underline{x}).
$$
Let $\star$ denote convolution. The family
$$
F_\varepsilon=F\star g_\varepsilon,
$$
called a Friedrichs mollifier,     
consists of smooth functions on $\R^{\bar{n}}$. They approximate $F$ in the supremum norm on $\R^{\bar{n}}$, 
that is we claim
\begin{equation}\label{claim:Tietze}
\lim_{\varepsilon\downarrow 0}\sup_{\underline{x}\in \R^{\bar{n}}} 
|F(\underline{x})-F_\varepsilon(\underline{x})|=0.
\end{equation}
and therefore {\it a fortiori} $f$ in the supremum norm on $\fC_{+,\kappa}$. Obviously $F_\varepsilon$ has 
support in the compact set consisting of points with a distance of at most $\varepsilon$ to the support of $F$.

To prove \eqref{claim:Tietze} we introduce $F(t)$ by 
$$
F(\underline{x},t)=F(\underline{x})
-\int_{\R^{\bar{n}}}\,\rd^{\bar{n}}\underline{y}\;F(\underline{x}-t\underline{y})\;g_\varepsilon(\underline{y}).
$$
Obviously $F(t=0)=0$ while $F(t=1)=F-F_\varepsilon$. Therefore the identities
\begin{align*}
F(\underline{x})-F_\varepsilon(\underline{x})&=F(\underline{x},t=1)
=\int_0^1 \;\rd s \frac{\rd}{\rd s}F(\underline{x}, s)\\
&= \int_0^1\rd s\;\int_{\R^{\bar{n}}}\,\rd ^{\bar{n}}\underline{y}\:
\;\underline{y}\cdot\underline{\nabla}\,F(\underline{x}-s\underline{y})\; g_\varepsilon(\underline{y})
\end{align*}
give the estimate 
\begin{align*}
|F(\underline{x})-F_\varepsilon(\underline{x})|&\le \varepsilon||\underline{\nabla}F||_{\sup}
\end{align*}
for all $\underline{x}\in\R^{\bar{n}}$. {\it A fortiori} the estimate 
$$
||F-F_{\varepsilon}||_{\sup}\le \varepsilon ||\underline{\nabla}F||_{\sup}
$$
is valid for all $0<\varepsilon <1$.
Use has been made of the facts that $|\underline{y}|\le \varepsilon$ for $\underline{y}$ lying in the support of 
$g_\varepsilon$, that $0\le g_\varepsilon$ and that
$$
\int_{\R^{\bar{n}}}\,\rd ^{\bar{n}}\underline{y}\; g_\varepsilon(\underline{y})=1.
$$
This concludes the proof of Lemma \ref{lem:Tietze}.
\end{appendix}


\end{document}